\documentclass[reqno,12pt,a4paper]{amsart}

\usepackage{mathrsfs}
\usepackage{amssymb}
\usepackage{amsfonts}
\usepackage{latexsym}
\usepackage{amsthm}
\usepackage{mathtools}

\usepackage{enumitem}
\setlist{nosep, leftmargin=*}
\usepackage{titletoc}
\usepackage{pict2e}

\numberwithin{equation}{section}

\usepackage{graphicx}
\usepackage{tikz}

\begin{document}


\renewcommand{\theequation}{\arabic{section}.\arabic{equation}}
\theoremstyle{plain}
\newtheorem{theorem}{\bf Theorem}[section]
\newtheorem{lemma}[theorem]{\bf Lemma}
\newtheorem{corollary}[theorem]{\bf Corollary}
\newtheorem{proposition}[theorem]{\bf Proposition}
\newtheorem{definition}[theorem]{\bf Definition}
\newtheorem*{definition*}{\bf Definition}
\newtheorem*{example}{\bf Example}
\newtheorem*{theorem*}{\bf Theorem}
\theoremstyle{remark}
\newtheorem*{remark}{\bf Remark}

\def\a{\alpha}  \def\cA{{\mathcal A}}     \def\bA{{\bf A}}  \def\mA{{\mathscr A}}
\def\b{\beta}   \def\cB{{\mathcal B}}     \def\bB{{\bf B}}  \def\mB{{\mathscr B}}
\def\g{\gamma}  \def\cC{{\mathcal C}}     \def\bC{{\bf C}}  \def\mC{{\mathscr C}}
\def\G{\Gamma}  \def\cD{{\mathcal D}}     \def\bD{{\bf D}}  \def\mD{{\mathscr D}}
\def\d{\delta}  \def\cE{{\mathcal E}}     \def\bE{{\bf E}}  \def\mE{{\mathscr E}}
\def\D{\Delta}  \def\cF{{\mathcal F}}     \def\bF{{\bf F}}  \def\mF{{\mathscr F}}
\def\c{\chi}    \def\cG{{\mathcal G}}     \def\bG{{\bf G}}  \def\mG{{\mathscr G}}
\def\z{\zeta}   \def\cH{{\mathcal H}}     \def\bH{{\bf H}}  \def\mH{{\mathscr H}}
\def\e{\eta}    \def\cI{{\mathcal I}}     \def\bI{{\bf I}}  \def\mI{{\mathscr I}}
\def\p{\psi}    \def\cJ{{\mathcal J}}     \def\bJ{{\bf J}}  \def\mJ{{\mathscr J}}
\def\vT{\Theta} \def\cK{{\mathcal K}}     \def\bK{{\bf K}}  \def\mK{{\mathscr K}}
\def\k{\kappa}  \def\cL{{\mathcal L}}     \def\bL{{\bf L}}  \def\mL{{\mathscr L}}
\def\l{\lambda} \def\cM{{\mathcal M}}     \def\bM{{\bf M}}  \def\mM{{\mathscr M}}
\def\L{\Lambda} \def\cN{{\mathcal N}}     \def\bN{{\bf N}}  \def\mN{{\mathscr N}}
\def\m{\mu}     \def\cO{{\mathcal O}}     \def\bO{{\bf O}}  \def\mO{{\mathscr O}}
\def\n{\nu}     \def\cP{{\mathcal P}}     \def\bP{{\bf P}}  \def\mP{{\mathscr P}}
\def\r{\varrho} \def\cQ{{\mathcal Q}}     \def\bQ{{\bf Q}}  \def\mQ{{\mathscr Q}}
\def\s{\sigma}  \def\cR{{\mathcal R}}     \def\bR{{\bf R}}  \def\mR{{\mathscr R}}
\def\S{\Sigma}  \def\cS{{\mathcal S}}     \def\bS{{\bf S}}  \def\mS{{\mathscr S}}
\def\t{\tau}    \def\cT{{\mathcal T}}     \def\bT{{\bf T}}  \def\mT{{\mathscr T}}
\def\f{\phi}    \def\cU{{\mathcal U}}     \def\bU{{\bf U}}  \def\mU{{\mathscr U}}
\def\F{\Phi}    \def\cV{{\mathcal V}}     \def\bV{{\bf V}}  \def\mV{{\mathscr V}}
\def\P{\Psi}    \def\cW{{\mathcal W}}     \def\bW{{\bf W}}  \def\mW{{\mathscr W}}
\def\o{\omega}  \def\cX{{\mathcal X}}     \def\bX{{\bf X}}  \def\mX{{\mathscr X}}
\def\x{\xi}     \def\cY{{\mathcal Y}}     \def\bY{{\bf Y}}  \def\mY{{\mathscr Y}}
\def\X{\Xi}     \def\cZ{{\mathcal Z}}     \def\bZ{{\bf Z}}  \def\mZ{{\mathscr Z}}
\def\O{\Omega}

\newcommand{\mc}{\mathscr {c}}

\newcommand{\gA}{\mathfrak{A}}          \newcommand{\ga}{\mathfrak{a}}
\newcommand{\gB}{\mathfrak{B}}          \newcommand{\gb}{\mathfrak{b}}
\newcommand{\gC}{\mathfrak{C}}          \newcommand{\gc}{\mathfrak{c}}
\newcommand{\gD}{\mathfrak{D}}          \newcommand{\gd}{\mathfrak{d}}
\newcommand{\gE}{\mathfrak{E}}
\newcommand{\gF}{\mathfrak{F}}           \newcommand{\gf}{\mathfrak{f}}
\newcommand{\gG}{\mathfrak{G}}           \newcommand{\Gg}{\mathfrak{g}}
\newcommand{\gH}{\mathfrak{H}}           \newcommand{\gh}{\mathfrak{h}}
\newcommand{\gI}{\mathfrak{I}}           \newcommand{\gi}{\mathfrak{i}}
\newcommand{\gJ}{\mathfrak{J}}           \newcommand{\gj}{\mathfrak{j}}
\newcommand{\gK}{\mathfrak{K}}            \newcommand{\gk}{\mathfrak{k}}
\newcommand{\gL}{\mathfrak{L}}            \newcommand{\gl}{\mathfrak{l}}
\newcommand{\gM}{\mathfrak{M}}            \newcommand{\gm}{\mathfrak{m}}
\newcommand{\gN}{\mathfrak{N}}            \newcommand{\gn}{\mathfrak{n}}
\newcommand{\gO}{\mathfrak{O}}
\newcommand{\gP}{\mathfrak{P}}             \newcommand{\gp}{\mathfrak{p}}
\newcommand{\gQ}{\mathfrak{Q}}             \newcommand{\gq}{\mathfrak{q}}
\newcommand{\gR}{\mathfrak{R}}             \newcommand{\gr}{\mathfrak{r}}
\newcommand{\gS}{\mathfrak{S}}              \newcommand{\gs}{\mathfrak{s}}
\newcommand{\gT}{\mathfrak{T}}             \newcommand{\gt}{\mathfrak{t}}
\newcommand{\gU}{\mathfrak{U}}             \newcommand{\gu}{\mathfrak{u}}
\newcommand{\gV}{\mathfrak{V}}             \newcommand{\gv}{\mathfrak{v}}
\newcommand{\gW}{\mathfrak{W}}             \newcommand{\gw}{\mathfrak{w}}
\newcommand{\gX}{\mathfrak{X}}               \newcommand{\gx}{\mathfrak{x}}
\newcommand{\gY}{\mathfrak{Y}}              \newcommand{\gy}{\mathfrak{y}}
\newcommand{\gZ}{\mathfrak{Z}}             \newcommand{\gz}{\mathfrak{z}}

\def\ve{\varepsilon}   \def\vt{\vartheta}    \def\vp{\varphi}    \def\vk{\varkappa}

\def\A{{\mathbb A}} \def\B{{\mathbb B}} \def\C{{\mathbb C}}
\def\dD{{\mathbb D}} \def\E{{\mathbb E}} \def\dF{{\mathbb F}} \def\dG{{\mathbb G}}
\def\H{{\mathbb H}}\def\I{{\mathbb I}} \def\J{{\mathbb J}} \def\K{{\mathbb K}} \def\dL{{\mathbb L}}
\def\M{{\mathbb M}} \def\N{{\mathbb N}} \def\O{{\mathbb O}} \def\dP{{\mathbb P}} \def\R{{\mathbb R}}
\def\dQ{{\mathbb Q}} \def\S{{\mathbb S}} \def\T{{\mathbb T}} \def\U{{\mathbb U}} \def\V{{\mathbb V}}
\def\W{{\mathbb W}} \def\X{{\mathbb X}} \def\Y{{\mathbb Y}} \def\Z{{\mathbb Z}}

\newcommand{\1}{\mathbbm 1}
\newcommand{\dd}    {\, \mathrm d}



\def\la{\leftarrow}              \def\ra{\rightarrow}            \def\Ra{\Rightarrow}
\def\ua{\uparrow}                \def\da{\downarrow}
\def\lra{\leftrightarrow}        \def\Lra{\Leftrightarrow}


\def\lt{\biggl}                  \def\rt{\biggr}
\def\ol{\overline}               \def\wt{\widetilde}
\def\no{\noindent}


\let\ge\geqslant                 \let\le\leqslant
\def\lan{\langle}                \def\ran{\rangle}
\def\/{\over}                    \def\iy{\infty}
\def\sm{\setminus}               \def\es{\emptyset}
\def\ss{\subset}                 \def\ts{\times}
\def\pa{\partial}                \def\os{\oplus}
\def\om{\ominus}                 \def\ev{\equiv}
\def\iint{\int\!\!\!\int}        \def\iintt{\mathop{\int\!\!\int\!\!\dots\!\!\int}\limits}
\def\el2{\ell^{\,2}}             \def\1{1\!\!1}
\def\sh{\sharp}
\def\wh{\widehat}
\def\ds{\dotplus}

\def\all{\mathop{\mathrm{all}}\nolimits}
\def\where{\mathop{\mathrm{where}}\nolimits}
\def\as{\mathop{\mathrm{as}}\nolimits}
\def\Area{\mathop{\mathrm{Area}}\nolimits}
\def\arg{\mathop{\mathrm{arg}}\nolimits}
\def\adj{\mathop{\mathrm{adj}}\nolimits}
\def\const{\mathop{\mathrm{const}}\nolimits}
\def\det{\mathop{\mathrm{det}}\nolimits}
\def\diag{\mathop{\mathrm{diag}}\nolimits}
\def\diam{\mathop{\mathrm{diam}}\nolimits}
\def\dim{\mathop{\mathrm{dim}}\nolimits}
\def\dist{\mathop{\mathrm{dist}}\nolimits}
\def\Im{\mathop{\mathrm{Im}}\nolimits}
\def\Iso{\mathop{\mathrm{Iso}}\nolimits}
\def\Ker{\mathop{\mathrm{Ker}}\nolimits}
\def\Lip{\mathop{\mathrm{Lip}}\nolimits}
\def\rank{\mathop{\mathrm{rank}}\limits}
\def\Ran{\mathop{\mathrm{Ran}}\nolimits}
\def\Re{\mathop{\mathrm{Re}}\nolimits}
\def\Res{\mathop{\mathrm{Res}}\nolimits}
\def\res{\mathop{\mathrm{res}}\limits}
\def\sign{\mathop{\mathrm{sign}}\nolimits}
\def\supp{\mathop{\mathrm{supp}}\nolimits}
\def\Tr{\mathop{\mathrm{Tr}}\nolimits}
\def\AC{\mathop{\rm AC}\nolimits}
\def\BBox{\hspace{1mm}\vrule height6pt width5.5pt depth0pt \hspace{6pt}}


\newcommand\nh[2]{\widehat{#1}\vphantom{#1}^{(#2)}}
\def\dia{\diamond}

\def\Oplus{\bigoplus\nolimits}




\def\qqq{\qquad}
\def\qq{\quad}
\let\ge\geqslant
\let\le\leqslant
\let\geq\geqslant
\let\leq\leqslant

\newcommand{\ca}{\begin{cases}}
\newcommand{\ac}{\end{cases}}
\newcommand{\ma}{\begin{pmatrix}}
\newcommand{\am}{\end{pmatrix}}
\renewcommand{\[}{\begin{equation}}
\renewcommand{\]}{\end{equation}}
\def\bu{\bullet}

\title[{}]{Stability of Dirac resonances}

\date{\today}

\author[Dmitrii Mokeev]{Dmitrii Mokeev}
\address{Saint Petersburg State University, Universitetskaya nab. 7/9, 
St. Petersburg, 199034, Russia, \ mokeev.ds@yandex.ru}


\begin{abstract}
	We prove that the class of resonances of Dirac operators on the half-line with 
	compactly supported potentials is closed with respect to $\ell^1$ perturbations. 
	We also prove that the potential depends continuously on such perturbations. 
	We show that similar results hold true for the Jost functions and Hermite-Biehler functions
	associated with Dirac operators.
\end{abstract}

\maketitle

\section{Introduction and main results}
\subsection{Introduction}
We consider the self-adjoint Dirac operator $H$ on $L^2(\R_+,\C^2)$
given by
$$
    H y = -i \s_3 y' + i \s_3 Q y,\qq y = \ma y_1 \\ y_2 \am,\qq \s_3 = \ma 1 & 0 \\ 0 & -1 \am,
$$
with the Dirichlet boundary condition
$$
	y_1(0) - y_2(0) = 0.
$$
The potential $Q$ has the following form
$$
    Q = \ma 0 & q \\ \overline{q} & 0 \am,\qq q \in \cP,
$$
where the class $\cP$ is defined for some $\g > 0$ fixed throughout
this paper by
\begin{definition*}
    $\cP = \cP_{\g}$ is a metric space of all functions $q \in L^2(\R_+)$ such that
	$\sup \supp q = \g$ equipped with the metric
	$$
		\rho_{\cP}(q_1,q_2) = \|q_1 - q_2\|_{L^2(0,\g)},\qq q_1,q_2 \in \cP.
	$$
\end{definition*}

It is well known that $\s(H) = \s_{ac}(H) = \R$ (see, e.g., \cite{LS91}).
For any $z \in \C$, we introduce the $2 \times 2$ matrix-valued Jost solution $f(x,z) =  \left(
\begin{smallmatrix} f_{11} & f_{12} \\ f_{21} & f_{22} \end{smallmatrix} \right) (x,z)$
of the Dirac equation
$$
    f'(x,z) = Q(x) f(x,z) + i z \s_3 f(x,z),\qq (x,z) \in \R_+ \ts \C,
$$
which satisfies the standard condition for compactly supported potentials:
$$
    f(x,z) = e^{i z x \s_3},\qq \forall \qq (x,z) \in [\g,+\iy) \ts \C.
$$
We define a Jost function $\psi: \C \to \C$ for the operator $H$ by
$$
    \psi(z) = f_{11} (0,z) - f_{21} (0,z),\qq z \in \C.
$$
It is well-known that $\psi$ is entire, $\psi(z) \neq 0$ for any $z
\in \ol \C_+ $ and it has zeros in $\C_-$, which are called
\textit{resonances} and a multiplicity of a resonance is a multiplicity of the zero of $\p$ 
(see, e.g., \cite{IK14b}).
Moreover, it was shown in \cite{IK14b} that the resonances are also zeros of the Fredholm
determinant and poles of the resolvent of the operator $H$.
We introduce a scattering matrix $S:\R \to \C$ for the operator $H$ by
$$
    S(z) = \frac{\ol\psi(z)}{\psi(z)} = e^{-2 i \arg \psi(z)},\qq z \in \R.
$$
The function $S$ admits a meromorphic continuation from $\R$ onto $\C$, since $\psi$ is entire.
Moreover, poles of $S$ are zeros of $\psi$ and then they are resonances.
Note that we sometimes write 
$\psi(\cdot,q)$, $S(\cdot,q)$,$\ldots$ instead of $\psi(\cdot)$, $S(\cdot)$, $\ldots$
when several potentials are being dealt with.

An inverse problem for operator $H$ consists of recovering the potential $q$ in specified
class by some spectral data. As spectral data one can consider, for instance, the 
scattering matrix $S$, the Jost function $\psi$ or resonances. The inverse problem for operator 
$H$ in terms of the scattering matrix is intensively studied in connection with the nonlinear
Schr{\"o}dinger equation (see, e.g., \cite{ZS71, APT04, DEGM82, FT07}). 
Inverse problems for the operator $H$ 
in terms of resonances and the Jost function have been studied much less. 
For the potentials $q \in \cP$ these problems were considered in \cite{KM20a}.
Moreover, in this paper, the characterization of the scattering matrix for $q \in \cP$
was given. In order to present these results, we introduce the Fourier transform $\cF$ on $L^2(\R)$ by
$$
	(\cF g)(k) = \int_{\R} g (s) e^{2iks} ds,\qq k \in \R.
$$
Then its inverse $\cF^{-1}$ on $L^2(\R)$ is given by
$$
	(\cF^{-1} g)(s) = \frac{1}{\pi} \int_{\R} g(k) e^{-2iks} dk,\qq s \in \R.
$$
We will use the notation $\hat g = \cF^{-1} g$. Note that if we apply the direct or inverse Fourier 
transform to a function $g$ defined on $I \ss \R$, then we extend $g$ to the whole line by making 
it zero outside $I$. 
Now, we introduce a class of Jost functions from \cite{KM20a}.
\begin{definition*}
	$\cJ = \cJ_{\g}$ is a metric space of all entire functions $\psi$ such that 
	\[	\label{p2e1}
		\psi(z) = 1  + \cF g(z),\qq z \in \C,
	\]
	for some $g \in \cP$ and $\psi(z) \neq 0$ for any $z \in \ol \C_+$
    equipped with the metric
    $$
        \rho_{\cJ}(\psi_1,\psi_2) = \| \cF^{-1}(\psi_1 - \psi_2) \|_{\cP},\qq \psi_1,\psi_2 \in \cJ.
    $$
\end{definition*}
We define the circle $\S^1 = \{\,z \in \C \, \mid \, |z| = 1 \, \}$.
Let $g:\R \to \S^1$ be a continuous function
such that $g(x) = C + o(1)$ as $x \to \pm \iy$ for some $C \in \S^1$.
Then $g = e^{-2i\phi}$ for some continuous $\phi:\R \to \R$.
We introduce a winding number $W(g) \in \Z$ by
$$
    W(g) = \frac{1}{\pi}\left(\lim_{x \to +\iy} \phi(x) - \lim_{x \to -\iy} \phi(x) \right),
$$
i.e., $W(g)$ is a number of revolutions of $g(x)$ around $0$, when
$x$ runs through $\R$. Now, we introduce class of scattering matrices from \cite{KM20a}.
\begin{definition*}
	$\cS = \cS_{\g}$ is a metric space of all continuous functions $S:\R \to \S^1$ such that:
	\begin{enumerate}[label={\roman*)}]
		\item $W(S) = 0$;
		\item $S(z) = 1  + \cF F(z)$, $z \in \R$, for some $F \in L^1(\R) \cap L^2(\R)$ such that $\inf \supp F = -\g$;
	\end{enumerate}
	equipped with the metric
    $$
        \rho_{\cS}(S_1,S_2) = \|\cF^{-1}(S_1 - S_2)\|_{L^2(-\g,+\iy)} + \|\cF^{-1}(S_1 - S_2)\|_{L^1(-\g,+\iy)},\qq
        S_1,S_2 \in \cS.
    $$
\end{definition*}

We need the following results from \cite{KM20a} about inverse problem 
for the operator $H$ in terms of the Jost function, the scattering matrix and resonances.
\begin{theorem} \label{thm:inv_jost}
	\begin{enumerate}[label={\roman*)}]
        \item The mapping $q \mapsto S(\cdot,q)$ from $\cP$ to $\cS$ is a homeomorphism;
        \item The mapping $q \mapsto \psi(\cdot,q)$ from $\cP$ to $\cJ$ is a homeomorphism;
		\item The potential $q \in \cP$ is uniquely determined by zeros of $\psi(\cdot,q) \in \cJ$.
    \end{enumerate}
\end{theorem}

This theorem solves the inverse problem in terms of resonances (\textit{uniqueness} and \textit{characterization})
in the following way: the potential $q \in \cP$ is uniquely determined by resonances, 
which are zeros of some $\psi \in \cJ$. That is, the sequence $(k_n)_{n \geq 1}$ such that 
$k_n \in \C_-$, $n \geq 1$, are the sequence of resonances for some $q \in \cP$ 
if and only if $(k_n)_{n \geq 1}$ are zeros of some $\psi \in \cJ$. Moreover, it was shown in 
\cite{KM20a}, how to recover the Jost function and the scattering matrix from resonances and then 
we can recover the potential using the Gelfand-Levitan-Marchenko equations.

The next question that can be posed is a \textit{stability} of this inverse problem. That is, how
the resonances for some $q_o \in \cP$ can be perturbed such that we also obtain the sequence of 
resonances for some $q \in \cP$. And related question is a \textit{continuity} of this inverse problem.
That is, does the potential depends continuously in some sense on the perturbation of resonances. 

In paper \cite{KM20a}, these problems were solved for the perturbation of a finite number of resonances. 
Namely, it was shown that if we have a sequence of resonances for 
some $q \in \cP$ and we arbitrarily shift a finite number of resonances, then we obtain 
the sequence of resonances which is associated with another potential from $\cP$ and 
the potential depends continuously on a distance between resonances.

Our main goal is to solve the global stability and continuity problems for the resonances of operator $H$, 
when infinitely many resonances are involved.

\subsection{Main result}
In order to formulate main result, we introduce the Banach space $\ell^1$ as 
a set of all sequences of complex numbers $\z = (\z_n)_{n \geq 1}$ equipped with the norm
$\displaystyle \|\z\|_{\ell^1} = \sum_{n \geq 1} |\z_n|$.
Let $\k = (k_n)_{n \geq 1}$ be a sequence of numbers from $\C_-$ such that 
$|k_1| \leq |k_2| \leq \ldots$. Then, by 
$q(\cdot,\k)$, we denote the potential such that $(k_n)_{n \geq 1}$ are its resonances, 
if such potential exists. Now, we give our main result.
\begin{theorem} \label{thm:main_thm}
	Let $\k^o = (k_n^o)_{n \geq 1}$ be zeros of $\psi(\cdot,q_o)$ for some $q_o \in \cP$ 
	arranged that $0 < |k^o_1| \leq |k^0_2| \leq \ldots$ and 
	let $\r = (\r_n)_{n \geq 1} \in \ell^1$ be such that $k_n = k_n^o + \r_n \in \C_-$ 
	for each $n \geq 1$. Then there exists a unique $q \in \cP$ such that 
	$\k = (k_n)_{n \geq 1}$ are zeros of $\psi(\cdot,q)$.
	Moreover, if $\|\r\|_{\ell^1} \to 0$, then we have $\|q-q_o\|_{\cP} \to 0$.
\end{theorem}

\begin{remark}
	1) Firstly, this theorem solve the global stability problem for resonances. 
	Namely, it shows that the set of resonances for $q \in \cP$ 
	is closed with respect to $\ell^1$ perturbations. Secondly, it shows
	that the potential depends continuously on such perturbations.

	2) In case of Schr{\"o}dinger operators on the half-line with compactly supported potentials, 
	the similar result was obtained by Korotyaev in \cite{K04b}. 
	It was shown in this paper that the space of resonances is closed under perturbations 
	$(\r_n)_{n \geq 1}$ such that $\sum_{n \geq 1} |\r_n|^2 n^{2\ve} < \iy$ for some $\ve > 1$.
\end{remark}

Albeit the methods of this paper are similar to those from \cite{K04b}, they need some 
adaptation. In particular, we used the methods from the Banach algebras theory.
This is due to the following differences between Dirac and Schr{\"o}dinger cases:

\begin{enumerate}[label={(\roman*)}]
    \item The resonances of Dirac operators are not symmetric with respect to the imaginary
    line.

    \item Roughly speaking, the spectral problem for Dirac operators
    corresponds to spectral problem for Schr{\"o}dinger operators with
    distributions.

    \item The second term in the asymptotic expansion of the Jost function of Dirac operators
    decrease more slowly as spectral parameter goes to infinity. Maybe it is the
    main point.
\end{enumerate}

The stability of inverse problem in terms of resonances for Schr{\"o}dinger operators 
on the half-line with compactly supported potentials was also considered by Marletta, 
Shterenberg and Weikard \cite{MSW10} in a different form. They showed that 
two potentials are close to each other if their resonances in the circle with a radius $R$ are close 
to each other. Namely, in this work, the norm 
$\sup_{x \in [0,\g]} \left| \int_{x}^{\g} (q(t) - \tilde{q}(t)) dt\right|$ was estimated
through $R$ and $\ve$, where $|z_n(q) - z_n(\tilde{q})| < \ve$ for any $n \geq 1$ such that
$|z_n(q)| < R$. Such results are possibly preferable for numerical applications since they answer on the
question how many resonances we need to know to recover the potential with a given accuracy.

An extension of this method for Schr{\"o}dinger operators 
on the real line with compactly supported potentials was obtained by Bledsoe in \cite{B12}. 
The inverse resonance problem in this case was studied by Korotyaev in \cite{K05}, where was shown
that in this case resonances does not uniquely determine a potential and then we need to add 
some additional data to obtain the uniqueness.

\subsection{Canonical systems}
It is well-known that the Dirac operators are associated with canonical systems (see, e.g., \cite{GK67}). 
We consider a canonical system given by
\[ \label{p1e7}
    Jy'(x,z) = z \gh(x) y(x,z),\qq (x,z) \in \R_+ \ts \C,\qq J = \ma 0 & 1 \\ -1 & 0 \am,
\]
with the Dirichlet boundary condition
\[ \label{p1e8}
	y_1(0,z) = 0,\qq y(x,z) = \ma y_1 \\ y_2 \am(x,z),
\]
where $\gh: \R_+ \to \cM^+_2(\R)$ is a Hamiltonian and by $\cM^+_2(\R)$, we denote the set of
$2 \times 2$ positive-definite self-adjoint matrices with real entries.
Now, we introduce the class of Hamiltonians associated with the Dirac operators.
By $\cM_2(\R)$ we denote the set of $2 \times 2$ matrices with real entries.
\begin{definition*}
    $\cG = \cG_{\g}$ is the set of functions $\gh:\R_+ \to \cM^+_2(\R)$ such that
    $$
        \gh' \in L^2(\R_+,\cM_2(\R)),\qq \sup \supp \gh' = \g,\qq \gh(0) = I_2
    $$
    and $\gh$ has the following form
    \[ \label{p1e12}
        \gh = \ma \ga & \gb \\ \gb & \frac{1+\gb^2}{\ga} \am,
    \]
    where $\ga : \R_+ \to \R_+$ and $\gb : \R_+ \to \R$.
\end{definition*}
\begin{remark}
    If $\gh \in \cG$, then it follows from (\ref{p1e12}) that
    $$
        \gh^*(x) = \gh(x),\qq \det \gh(x) = 1,\qq \ga(x) > 0,\qq \forall x \in \R_+
    $$
    and $\gh(x)$ is a constant matrix for any $x \geq \g$.
\end{remark}
The canonical system (\ref{p1e7}), (\ref{p1e8}) with the Hamiltonian $\gh \in \cG$ 
corresponds to an self-adjoint operator
$$
	\cK_{\gh} = \gh^{-1} J \frac{d}{dx}
$$
in the weighted Hilbert space $L^2(\R_+, \C^2, \cH)$ equipped with the norm
$$
    \|f\|^2_{L^2(\R_+, \C^2, \gh)} = \int_{\R_+} (\gh(x) f(x), f(x)) dx,\qq f \in L^2(\R_+, \C^2, \gh),
$$
where $(\cdot,\cdot)$ is the standard scalar product in $\C^2$ (see, e.g., \cite{R14}).
Moreover, we need the following result from \cite{KM20a} (see also \cite{KM20b}).
\begin{theorem} \label{thm3}
	For any $q \in \cP$ there exists a unique $\gh \in \cG$ and for any $\gh \in \cG$
	there exists a unique $q \in \cP$ such that the operators $H(q)$ and $\cK(\gh)$ are
	unitary equivalent.
\end{theorem}
Recall well-known result about inverse problem for the canonical system (\ref{p1e7}), (\ref{p1e8})
in terms of de Branges space (see Theorem 40 in \cite{dB} or Theorems 10, 13 in \cite{R14}).
For any canonical system there exists a Hermite-Biehler function $E$ such that $E$ is entire and
$|E(z)| > |E(\ol z)|$ for each $z \in \C_+$ and the associated de Branges space $B(E)$ is given by
$$
    B(E) = \Big\{ F:\C \to \C \, \mid \, \text{$F$ is entire},\, \frac{F}{E},\,
    \frac{F^{\#}}{E} \in \mH^2(\C_+)\, \Big\},
$$
where $F^{\#}(z) = \ol{F(\ol z)}$ and $\mH^2(\C_+)$ is the Hardy space in the upper half-plane.
Moreover, from a de Branges space, one can recover the associated canonical system.

We say that a Hermite-Biehler function is \textit{Dirac-type} if
it is associated with the canonical system with the Hamiltonian $\gh \in \cG$.
It follows from Theorem \ref{thm3} that there exists a correspondence between 
Dirac-type Hermite-Biehler functions and Jost functions. We need the following result
from \cite{KM20a}.
\begin{theorem} \label{thm4}
    A Hermite-Biehler function $E$ is Dirac-type if and only if
    \[ \label{p1e10}
        E(k) = -ie^{-i\g k} \psi(k),\qq k \in \C,
    \]
    for some $\psi \in \cJ$.
\end{theorem}

\begin{remark}
	Recall that Jost function is uniquely determined by its zeros. Thus, it follows from (\ref{p1e10})
	that Dirac-type Hermite-Biehler function is also uniquely determined by its zeros. Moreover,
	other properties of zeros of Jost functions hold true for Dirac-type Hermite-Biehler functions,
	see details in \cite{KM20a}.
\end{remark}
Now, using this correspondence, we show that the zeros of the Dirac-type Hermite-Biehler function
can be perturbed by $\ell^1$ sequence and the function depends continuously on this perturbation. 
\begin{theorem} \label{thm:herm_pert}
	Let $E_o$ be a Dirac-type Hermite-Biehler function with zeros 
	$\k^o = (k_n^o)_{n \geq 1}$ arranged that $0 < |k^o_1| \leq |k^0_2| \leq \ldots$ and 
	let $\r = (\r_n)_{n \geq 1} \in \ell^1$ be such that $k_n = k_n^o + \r_n \in \C_-$ 
	for each $n \geq 1$. Then there exists a unique Dirac-type Hermite-Biehler function $E$ 
	such that $\k = (k_n)_{n \geq 1}$ are zeros of $E$. 
	Moreover, if $\|\r\|_{\ell^1} \to 0$, then we have $\|\cF^{-1}(E - E_o)\|_{L^2(-\frac{\g}{2},\frac{\g}{2})} \to 0$.
\end{theorem}
\begin{remark}
	Note that it follows from the Plancherel theorem (see, e.g., Theorem IX.6 in \cite{RS80}) that
	$\|\cF^{-1}(E - E_o)\|_{L^2(-\frac{\g}{2},\frac{\g}{2})} \to 0$ if and only if
	$\|E - E_o\|_{L^2(\R)} \to 0$
\end{remark}

\subsection{Literature survey}
We shortly discuss the known results about resonances.
Resonances are considered in the different settings,
see articles \cite{F97, H99, K04a, S00, Z87} and the book
 \cite{DZ19} and the references therein. Recall that the inverse resonance problem for
Schr{\"o}dinger operators with compactly supported potentials was
solved by Korotyaev in \cite{K05} for the case of the real line and in
\cite{K04a} for the case of the half-line, 
see also Zworski \cite{Z02} and Brown, Knowles and Weikard
\cite{BKW03} concerning the uniqueness.
Moreover, there are other results about
perturbations of the following model (unperturbed) potentials by
compactly supported potentials: step potentials was considered by Christiansen \cite{C06}, 
periodic and linear potentials was considered by Korotyaev in \cite{K11h} and \cite{K17}.
Note also that Schr{\"o}dinger operators with linear potentials are one-dimensional Stark operators 
and in \cite{K17} the inverse resonance problem for Stark operators perturbed by compactly
supported potentials was solved.
The asymptotics of the counting function of resonances for Schr{\"o}dinger
operators on the real line with compactly supported potentials 
was first obtained by Zworski in \cite{Z87}. The 
results about the Carleson measures for resonances
were obtained by Korotyaev in \cite{K16}.

Global estimates of resonances for the massless Dirac operators on the real line were
obtained by Korotyaev in \cite{K14}. Resonances for Dirac operators was also
studied by Iantchenko and Korotyaev in \cite{IK14b} for the massive Dirac operators on the
half-line and in \cite{IK14a} for the massless Dirac operators on
the real line under the condition $q' \in L^1(\R)$.
In \cite{IK15}, Iantchenko and Korotyaev considered the radial Dirac operator.
The inverse resonance problem for the massless Dirac operators with 
compactly supported potentials was solved by Korotyaev and Mokeev in \cite{KM20a} 
for the case of the half-line and in \cite{KM20b} for the case of the real line.
There is a number of papers dealing with other related problems for the
one-dimensional Dirac operators, for instance, the resonances for
Dirac fields in black holes was described, see, e.g.,, Iantchenko \cite{I18}.

As we have showed above, Dirac operators can be rewritten as canonical systems and
for these systems, the inverse problem
can be solved in terms of de Branges spaces,
which can be parametrized by the Hermite-Biehler functions (see \cite{dB, R14}).
There exist many papers devoted to de Branges spaces and canonical
systems. In particular, they are used in the inverse spectral theory
of Schr{\"o}dinger and Dirac operators (see, e.g., \cite{R02}). In our paper we
have used the connection between Jost and Hermite-Biehler functions. 
Similar connection in case of the Schr{\"o}dinger operators was given
by Baranov, Belov and Poltoratski in \cite{BBP} (see also Makarov and Poltoratski \cite{P}).
In \cite{KM20a,KM20b}, the canonical systems associated with Dirac operators on the half-line
and on the real line was considered. In particular, it was shown how to recover the
potential of the Dirac operator by the Hamiltonian of the unitary equivalent canonical system.

\section{Preliminary}

Before we prove the main theorem, we recall some well-known facts about entire functions and 
Banach algebras, prove several technical lemmas and recall properties of the resonances of 
the operator $H$.

\subsection{Entire functions}
Recall that an entire function $f(k)$ is said to be of
\textit{exponential type} if there exist constants $\t,C > 0$ such that $|f(k)| \leq C e^{\t |k|}$,
$k \in \C$. We introduce the Cartwright classes of entire functions $\cE_{Cart}(\a,\b)$ by
\begin{definition*}
    For any $\a,\b \in \R$, $\cE_{Cart}(\a,\b)$ is a class of entire functions of exponential type $f$ such that
    $$
        \int_{\R} \frac{\log(1+|f(k)|)dk}{1 + k^2} < \iy,\qq \t_+(f) = \a,\qq \t_-(f) = \b,
    $$
	where $\displaystyle \t_{\pm}(f) = \lim \sup_{r \to +\iy} \frac{\log |f(\pm i r)|}{r}$. Let also
	$\cE_{Cart} = \cE_{Cart}(0,2\g)$.
\end{definition*}
If $f \in \cE_{Cart}(\a,\b)$ for some $\a,\b \in \R$, then it has the Hadamard factorization 
(see, e.g., pp.127-130 in \cite{L96}).
Let $p \geq 0$ be the multiplicity of zero $k = 0$ of $f$.
We denote by $(k_n)_{n \geq 1}$ zeros of $f$ in $\C \sm \{ 0 \}$
counted with multiplicity and arranged that $0 < |k_1| \leq |k_2| \leq \ldots$. 
Then $f$ has the Hadamard factorization
\[ \label{p2e11}
    f(k) = C k^p e^{i \varkappa k} \lim_{r \to +\iy}
        \prod_{|k_n| \leq r} \left(1 - \frac{k}{k_n}\right),\qq k \in \C,
\]
where the product converges uniformly
on compact subsets of $\C$ and
\[ \label{p2e12}
    \varkappa = \frac{\b - \a}{2},\qq C = \frac{f^{(p)}(0)}{p!},\qq
	\sum_{n \geq 1} \frac{|\Im k_n|}{|k_n|^2} < +\iy,\qq
	\exists \lim_{r \to +\iy} \sum_{|k_n| \leq r} \frac{1}{k_n} \neq \iy.
\]
For functions from $\cE_{Cart}(\a,\b)$ there is the Levinson's theorem about distribution of their zeros
(see, e.g.,  p. 58 in \cite{Koo98}). For any $f \in \cE_{Cart}(\a,\b)$, $\a,\b \in \R$, we introduce
the counting function of its zeros $n(r,f)$ into a circle of radius $r \geq 0$ by
$$
	n(r,f) = \# \{ \, k \in \C \, \mid \,f(k) = 0,\, |k| \leq r \,\}.
$$
\begin{theorem}[Levinson] \label{thm:lev}
    Let $f \in \cE_{Cart}(\a,\b)$ for some $\a,\b \in \R$. Then we have
    \[ \label{p2e13}
        n(r,f) = \frac{\a + \b}{\pi} r + o(r)
    \]
    as $r \to +\iy$.
\end{theorem}
We also need the following Lindel{\"o}f's theorem (see, e.g., p. 21 in \cite{Koo98}).
\begin{theorem}[Lindel{\"o}f] \label{thm:lind}
	Let $(k_n)_{n \geq 1}$ be arranged that $0 < |k_1| \leq |k_2| \leq |k_3| \leq \ldots$ and let
	$$
		n(r) = \# \{\, k_m,\, m \geq 1\,  \mid \,  |k_m| \leq r \,\}.
	$$
	Suppose that $n(r) \leq Kr$ for some $K \geq 0$ and for any $r \geq 0$ and suppose that
	$$
		\Bigl| \sum_{|k_n| \leq r} \frac{1}{k_n} \Bigr|
	$$
	remain bounded as $r \to \iy$. Then the product
	$$
		C(k) = \lim_{r \to +\iy} \prod_{|k_n| \leq r} \left(1 -\frac{k}{k_n}\right),\qq k \in \C,
	$$
	is an entire function of exponential type.
\end{theorem}
\begin{remark}
	Using this theorem, we can construct the entire functions of exponential type by its zeros.
	Note that the Lindel{\"o}f theorem is usually formulated for the canonical product of the form
	$$
		C_1(k) = \prod_{n \geq 1} \left(1 -\frac{k}{k_n}\right) e^{\frac{k}{k_n}},\qq k \in \C.
	$$
	We have replaced $C_1$ by $C$ using the standard arguments (see, e.g., p. 130 in \cite{L96}).
\end{remark}
We also need the following simple lemma about asymptotic of the counting function of sequence with
bounded perturbation. 
\begin{lemma} \label{lm:count_func}
	Let the sequences $(k_n)_{n \geq 1}$ and $(k^o_n)_{n \geq 1}$ be arranged that
	$$
		0 < |k_1| \leq |k_2| \leq |k_3| \leq \ldots,\qq 0 < |k^o_1| \leq |k^o_2| \leq |k^o_3| \leq \ldots.
	$$
	and let
	$$
		n(r) = \# \{\, k_m,\, m \geq 1\,  \mid \,  |k_m| \leq r \,\},\qq 
		n_o(r) = \# \{\, k^o_m,\, m \geq 1\,  \mid \,  |k^o_m| \leq r \,\}.
	$$
	Suppose also that
	\begin{enumerate}[label = {(\roman*)}]
		\item $\displaystyle \sup_{n \geq 1} |k_n - k_n^o| = s < \iy$,
		\item $n_o(r) = Cr + o(r)$ as $r \to \iy$ for some $C \in \R$.
	\end{enumerate}
	Then we have $n(r) = Cr + o(r)$ as $r \to \iy$.
\end{lemma}
\begin{proof}
	Since $\displaystyle \sup_{n \geq 1} |k_n - k_n^o| = s$, we have
	$$
		n_o(r-2s) \leq n(r) \leq n_o(r+2s),
	$$
	for any $r > 0$.
	Using $n_o(r) = Cr + o(r)$ as $r \to \iy$, we get
	$$
		C(r-2s) + o(r) \leq n(r) \leq C(r+2s) + o(r)
	$$
	as $r \to \iy$, which yields that $n(r) = Cr + o(r)$ as $r \to \iy$.
\end{proof}

\subsection{Resonances}
Recall that resonances of the operator $H$ are zeros of the associated Jost function, which is 
entire function of exponential type. Moreover, using the Paley-Wiener theorem 
(see, e.g., p.30 in \cite{Koo98}), we have that an entire function having form (\ref{p2e1}) 
belongs to the Cartwright class. Thus, the resonances and the Jost function of the operator $H$
have all properties, which was discussed above, and we have the following corollary
(see Corollary 1.2 in \cite{KM20a}). 
\begin{corollary} \label{cor:jost_cart}
    Let $q \in \cP$. Then we have $\psi(\cdot,q) \in \cE_{Cart}$ and it satisfies (\ref{p2e11}-\ref{p2e13}).
\end{corollary}

We also need the following result about position of resonances (see Theorem 1.3 in \cite{KM20a}).
\begin{theorem} \label{thm:res_pos}
    Let $q \in \cP$ and let $(k_n)_{n \geq 1}$ be its resonances. Let $\ve > 0$. Then there exists
    a constant $C = C(\ve,q) \geq 0$ such that the following inequality holds true for each $n \geq 1$:
    \[ \label{p1e1}
        2 \g \Im k_n \leq \ln \left( \ve + \frac{C}{|k_n|} \right).
    \]
    In particular, for any $A > 0$, there are only finitely many resonances in the strip
    \[ \label{p1e2}
        \{ \, k \in \C \, \mid \, 0 > \Im k > -A \, \}.
    \]
\end{theorem}
\begin{remark}
	This theorem describe so called forbidden domain for resonances. Moreover,
    if $q' \in L^1(\R_+)$, then estimate (\ref{p1e1}) and the forbidden domain (\ref{p1e2})
    can be given in more detailed form (see Theorem 2.7 in \cite{IK14b}).
\end{remark}
Using Theorem \ref{thm:res_pos}, we obtain the following useful corollary.
\begin{corollary} \label{cor:im_part}
	Let $q \in \cP$ and let $(k_n)_{n \geq 1}$ be its resonances. Then we have $\Im k_n \to -\iy$
	as $n \to \iy$.
\end{corollary}

\subsection{Banach algebras}
Recall that we have introduced the Fourier transform $\cF$ and its inverse $\cF^{-1}$ on $L^2(\R)$ by
$$
	\begin{aligned}
		(\cF g)(k) &= \int_{\R} g (s) e^{2iks} ds,\qq k \in \R,\\
		\cF^{-1} g(s) &= \frac{1}{\pi} \int_{\R} g(k) e^{-2iks} dk,\qq s \in \R.
	\end{aligned}	
$$
Moreover, we have introduced the notation $\hat g = \cF^{-1} g$. We introduce the following Banach space
$$
    \begin{aligned}
        \cL_{+} &= L^2(\R_{+}) \cap L^1(\R_{+}),\qq
        \| \cdot \|_{\cL_{+}} = \| \cdot \|_{L^2(\R_{+})} + \| \cdot \|_{L^1(\R_{+})}
	\end{aligned}
$$
and the following Banach algebras with pointwise multiplication
$$
    \begin{aligned}
		\hat\cL_{+} &= \{ \,\cF g \, \mid \, g \in \cL_{+} \,\},\qq
        \| \cF g \|_{{\hat \cL}_{+}} = \| g \|_{\cL_{+}},\\
        {\cW}_{+} &= \{ \, c + g \, \mid \, (c,g) \in \C \ts \hat \cL_{+} \,\},\qq
        \| c + g \|_{{\cW}_{+}} = |c| + \| g \|_{\hat \cL_{+}}.
    \end{aligned}
$$
It is well-known that ${\cW}_+$ is unital Banach algebra (see, e.g., Chapter 17 in \cite{GRS64}).
Moreover, due to Paley-Wiener theorem and the Riemann-Lebesgue lemma (see, e.g., Theorem IX.7 in \cite{RS80}),
each element of $\hat \cL_+$ or $\cW_+$ is bounded continuous function on $\ol \C_+$.
The spectrum of $f \in \cW_+$ is given by
$$
	\s(f) = \{\, f(k) \,\mid\, k \in \ol \C_+ \,\} \cup \{\,\lim_{k \to \iy} f(k)\,\}.
$$
Thus, $f \in \cW_+$ is invertible in $\cW_+$ if and only if $f(k) \neq 0$ for any 
$k \in \ol \C_+$ and $\lim_{k \to \iy} f(k) \neq 0$. Recall that in each Banach algebra there exists
a holomorphic functional calculus, that is the following theorem holds true (see, e.g., Chapter 6 in \cite{GRS64}).
\begin{theorem} \label{thm5}
	Let $\vp$ be an analytic function on some open domain $D$ and let $f \in \cW_+$ such that
	$\s(f) \ss D$. Suppose that $\G \ss D$ is a closed rectifiable curve such that $\s(x)$ is 
	contained in the interior of the domain bounded by $\G$. 
	Then there exists a unique $\vp(f) \in \cW_+$ given by
	$$
		\vp(f) = \frac{1}{2\pi i} \int_{\G} (\l - f)^{-1} \vp(\l) d\l,
	$$
	where the integral does not depend on the choice of $\G$, subject only to the conditions stated,
	and $\vp(f)$ depends continuously on $f \in \cW_+$ such that $\s(f)$ is 
	contained in the interior of the domain bounded by $\G$.
\end{theorem}
\begin{remark}
	The continuity of this mapping follows from the resolvent identity
	$$
		(\l - f_1)^{-1} - (\l - f_2)^{-1} = (\l - f_1)^{-1}(f_1 - f_2)(\l - f_2)^{-1},\qq f_1,f_2 \in \cW_+.	
	$$
\end{remark}
Using Theorem \ref{thm5}, we can consider analytic functions on $\cW_+$.
In particular, we introduce the logarithm and the exponential mappings on some subspaces of $\cW_+$. 
We introduce the following subspaces of $\cW_+$:
$$
	\begin{aligned}
		\cW_{log} &= \{\, f = 1 + g \, \mid \, g \in \hat \cL_+,\, f(x) \notin \R_- \cup \{0\},\, x \in \ol \C_+ \, \},\\
		\cW_{exp} &= \{\, f = 1 + g \, \mid \, g \in \hat \cL_+,\, |f(x)| > 0,\, x \in \ol \C_+ \, \}.
	\end{aligned}
$$
Let $\exp: f \mapsto e^{f(\cdot)}$, $f \in \hat \cL$, and
$\log : f \mapsto \log(f(\cdot))$, where 
the branch of the logarithm is fixed by the condition $\log(x) \in \R$ for any $x \in \R$.
Thus, we get the following corollary of Theorem~\ref{thm5}.
\begin{corollary} \label{p3c1}
	The mappings $\log : \cW_{log} \to \hat \cL_+$ and $\exp : \hat \cL_+ \to \cW_{exp}$
	are continuous.
\end{corollary} 

Moreover, we estimate the norm of the logarithm mapping.
\begin{lemma} \label{lm:log_estimate}
	Let $f \in \hat \cL_+$ be such that $\| f \|_{\hat \cL_+} < \frac{1}{4}$. Then we have
	$\log(1 + f) \in \hat \cL_+$ and
	$$
		\| \log (1+f) \|_{\hat \cL_+} < C \| f \|_{\hat \cL_+},
	$$
	there $C > 0$ does not depend on $f$.
\end{lemma}
\begin{proof}
	Let $\|f\|_{\hat \cL} = r < \frac{1}{4}$ and let $\G = \{\, z \in \C \, \mid \, |z| = 2r \,\}$. 
	Since $|\l| \leq \|f\|_{\hat \cL}$ for any $\l \in \s(f)$, we have that $\s(f)$ is 
	contained in the interior of the domain bounded by $\G$. 
	Recall that the analytic branch of the logarithm $\log(\cdot)$ on $\C \sm (-\iy,0]$ 
	is fixed by the condition $\log(z) \in \R$ for any $z > 0$. 
	Due to $r < \frac{1}{4}$, we have $1 + \l \in \C \sm (-\iy,0]$ for any $\l \in \G$. 
	Thus, using Theorem \ref{thm5}, we have
	$$
		\log(1 + f) = \frac{1}{2\pi i} \int_{\G} (\l - f )^{-1} \log(1 + \l) d\l,
	$$
	which yields
	\[ \label{log_estimate:eq1}
		\|\log(1 + f)\|_{\hat \cL_+} = \frac{1}{2\pi} \int_{\G} \|(\l - f )^{-1}\|_{\hat \cL_+} |\log(1 + \l)| |d\l|.
	\]
	Firstly, we estimate $\|(\l - f )^{-1}\|_{\hat \cL_+}$. Since $|\l| = 2r$ for any $\l \in \G$, we have
	\[ \label{log_estimate:eq2}
		\|(\l - f )^{-1}\|_{\hat \cL_+} = \frac{1}{|\l|} \Bigl\|\left(1 - \frac{f}{\l} \right)^{-1}\Bigr\|_{\hat \cL_+}
		\leq \frac{1}{|\l|} \frac{1}{1 - \frac{\left\| f\right\|}{|\l|}} = \frac{1}{2r} \frac{1}{1 - \frac{r}{2r}} = \frac{1}{r},
		\qq \l \in \G.
	\]
	Secondly, we estimate $|\log(1 + \l)|$. Since $\log(1+\l) = \l + o(\l)$ as $\l \to 0$, there exists
	a constant $C > 0$ such that
	\[ \label{log_estimate:eq3}
		|\log(1 + \l)| \leq C|\l|
	\]
	for any $\l \in \C$ such that $|\l| < 1/2$.
	Substituting (\ref{log_estimate:eq2}) and (\ref{log_estimate:eq3}) in (\ref{log_estimate:eq1}) and
	using $|\l| = 2r$ for any $\l \in \G$, we get
	$$
		\|\log(1 + f)\|_{\hat \cL_+} \leq \frac{1}{2\pi} \int_{\G} \frac{1}{r} C r |d\l| = 
		\frac{C}{2\pi} \int_{\G} |d\l| = 2Cr = 2C\| f \|_{\hat \cL_+}.
	$$
\end{proof}

\section{Proof of the main theorem}
Firstly, we consider simple function, which will be used in the proof of the main theorem. 
\begin{lemma} \label{lm:wiener_estimate}
	Let $k_1,k_2 \in \C_-$, $\r = k_2 - k_1$ and let 
	$$
		g(k) = 1 + \frac{\r}{k_1 - k},\qq k \in \C \sm \{k_1\}.
	$$
	Then we have $g \in \cW_+$, $g(k) \notin (-\iy,0]$ for any $k \in \ol \C_+$ and
	$$
		\|g-1\|_{\hat \cL_+} = \frac{|\r|}{|\Im k_1|^{\frac{1}{2}}}\left(1 + \frac{1}{|\Im k_1|^{\frac{1}{2}}} \right).
	$$
\end{lemma}
\begin{proof}
	Firstly, we show that $g(k) \notin (-\iy,0]$ for any $k \in \ol \C_+$ by contradiction. Let
	$$
		g(k) = 1 + \frac{\r}{k_1 - k} = \frac{k_2-k}{k_1 - k} = -c
	$$
	for some $c \in [0,+\iy)$ and $k \in \ol \C_+$. Then we have
	$$
		k_2 - k = -c(k_1-k).
	$$
	Considering the imaginary part of this identity and using $c \in \R$, we get
	$$
		\Im k_2 - \Im k = -c(\Im k_1 - \Im k).
	$$
	Due to $c, \Im k \geq 0$ and $\Im k_1, \Im k_2 < 0$, we have
	$$
		\Im k_2 - \Im k < 0,\qq -c(\Im k_1 - \Im k) \geq 0,
	$$
	which yields
	$$
		\Im k_2 - \Im k \neq -c(\Im k_1 - \Im k)
	$$
	and then we have a contradiction.
	
	Secondly, we show that $g \in \cW_+$. Using the Jordan lemma and $k_1 \in \C_-$, we obtain
	$$
		\cF^{-1} (g-1)(s) = \frac{\r}{\pi} \int_{\R} \frac{e^{-2iks}}{k_1 - k} dk = 2i\r e^{-2ik_1s}\vt(s),\qq s \in \R.
	$$
	Let $k_1 = x_1 + iy_1$. Then we have
	\[ \label{wiener_estimate:eq1}
		|\cF^{-1} (g-1)(s)| = 2|\r|e^{2y_1s} \vt(s),\qq s \in \R.
	\]
	Using (\ref{wiener_estimate:eq1}) and the fact that $y_1 < 0$, we get
	\[ \label{wiener_estimate:eq2}
		\|\cF^{-1} (g-1)\|_{L^1(\R)} = 2|\r|\int_0^{\iy} e^{2y_1s} ds = \frac{|\r|}{|y_1|}.
	\]
	Similarly we obtain
	\[ \label{wiener_estimate:eq3}
		\|\cF^{-1} (g-1)\|_{L^2(\R)} = 2|\r|\left|\int_0^{\iy} e^{4y_1s} ds\right|^{\frac{1}{2}} = 
		\frac{|\r|}{|y_1|^{\frac{1}{2}}}.
	\]
	Combining (\ref{wiener_estimate:eq2}) and (\ref{wiener_estimate:eq3}), we have
	$$
		\|g-1\|_{\hat \cL_+} = \frac{|\r|}{|y_1|^{\frac{1}{2}}}\left(1 + \frac{1}{|y_1|^{\frac{1}{2}}} \right).
	$$
\end{proof}

Secondly, we show the stability of Jost functions in $\cJ$ under $\ell^1$ perturbation of 
their zeros.
\begin{theorem} \label{thm:jost_pert}
	Let $f_o \in \cJ$ with zeros $(k_n^o)_{n \geq 1}$ arranged that 
	$0 < |k_1| \leq |k_2| \leq \ldots$ and let $\r = (\r_n)_{n \geq 1} \in \ell^1$ 
	be a sequence of complex numbers such that $k_n = k_n^o + \r_n \in \C_-$ for any $n \geq 1$. Then 
	there exists a unique $f \in \cJ$ such that $(k_n)_{n \geq 1}$ are its zeros. Moreover, if 
	$\| \r \|_{\ell^1} \to 0$, then we have $\r_{\cJ}(f,f_o) \to 0$.
\end{theorem}
\begin{proof}
	We show that there exists an entire function, which zeros are $(k_n)_{n \geq 1}$.
	Let 
	$$
		n(r) = \# \{\, k_m,\, m \geq 1\,  \mid \,  |k_m| \leq r \,\}.
	$$
	Due to Theorem \ref{cor:jost_cart}, we have $n_o(r) = \frac{2\g}{\pi} r + o(r)$ as $r \to \iy$.
	Since $\r \in \ell^1$, we have
	$$
		\sup_{m \geq 1} (k_m -k_m^o) = \sup_{m \geq 1} \r_m = s < \iy.
	$$
	Thus, using Lemma \ref{lm:count_func}, we get $n(r) = \frac{2\g}{\pi} r + o(r)$ as $r \to \iy$.
	
	Now, we show that $\displaystyle \Bigl| \sum_{|k_n| \leq r} \frac{1}{k_n} \Bigr|$ is bounded 
	as $r \to \iy$.
	Firstly, we consider the imaginary part of this sum. Using $k_n = k^o_n + \r_n$, we get
	\[ \label{jost_pert:eq1}
		\begin{aligned}
			\Bigl| \Im \sum_{|k_n| \leq r} \frac{1}{k_n} \Bigr| &\leq \sum_{|k_n| \leq r} \frac{|\Im k_n|}{|k_n|^2} 
			= \sum_{|k_n| \leq r} \frac{|\Im k_n^o|}{|k_n^o|^2} \frac{|\Im k_n|}{|\Im k_n^o|}\frac{|k_n^o|^2}{|k_n|^2}\\
			&= \sum_{|k_n| \leq r} \frac{|\Im k_n^o|}{|k_n^o|^2} \left| 1 + \frac{\Im \r_n}{\Im k_n^o}\right|
			\left|1 + \frac{\r_n}{k_n^o}\right|^{-2} = \sum_{|k_n| \leq r} \frac{|\Im k_n^o|}{|k_n^o|^2} \z_n,
		\end{aligned}
	\]
	where we have introduced $\displaystyle \z_n = \left| 1 + \frac{\Im \r_n}{\Im k_n^o}\right|
	\left|1 + \frac{\r_n}{k_n^o}\right|^{-2}$, $n \geq 1$.
	Due to Corollary \ref{cor:im_part}, we have $|k_n^o| \to \iy$ and $|\Im k_n^o| \to \iy$ as $n \to \iy$. 
	Recall that $\displaystyle \sup_{n \geq 1} \r_n = s < \iy$. Hence, we obtain
	$$
		\left|\frac{\Im \r_n}{\Im k_n^o}\right| \to 0,\qq \left|\frac{\r_n}{k_n^o}\right| \to 0
	$$
	as $n \to \iy$ and then there exists $C \in \R$ such that $|\z_n| < C$ for any $n \geq 1$.
	Substituting this estimate in (\ref{jost_pert:eq1}), we get
	$$
		\Bigl| \Im \sum_{|k_n| \leq r} \frac{1}{k_n} \Bigr| < C \sum_{|k_n| \leq r} \frac{|\Im k_n^o|}{|k_n^o|^2}.
	$$
	Using (\ref{p2e12}) and Corollary \ref{cor:jost_cart}, we have $\displaystyle \sum_{n \geq 1} \frac{|\Im k_n^o|}{|k_n^o|^2} < \iy$ and then the sum
	$\displaystyle \sum_{|k_n| \leq r} \frac{|\Im k_n^o|}{|k_n^o|^2}$ is bounded as $r \to \iy$. Thus,
	the sum $\displaystyle \Bigl| \Im \sum_{|k_n| \leq r} \frac{1}{k_n} \Bigr|$ is bounded as $r \to \iy$.
	
	Secondly, we consider the real part of this sum.
	$$
		\Bigl| \Re \sum_{|k_n| \leq r} \frac{1}{k_n} \Bigr| \leq \Bigl| \Re \sum_{|k_n| \leq r} \frac{1}{k^o_n} \Bigr| + \sum_{|k_n| \leq r} \left| \frac{1}{k_n} - \frac{1}{k_n^o} \right|.
	$$
	Due to (\ref{p2e12}) and Corollary \ref{cor:jost_cart}, the sum 
	$\displaystyle \sum_{|k_n| \leq r} \frac{1}{k^o_n}$ converges as $r \to \iy$ and then
	$\displaystyle \Bigl| \Re \sum_{|k_n| \leq r} \frac{1}{k^o_n} \Bigr|$ is bounded as $r \to \iy$.
	 Using $k_n = k_n^o + \r_n$, we get
	$$
		\sum_{|k_n| \leq r} \left| \frac{1}{k_n} - \frac{1}{k_n^o} \right| = \sum_{|k_n| \leq r} \frac{|\r_n|}{|k_n||k_n^o|} 
		\leq \sum_{|k_n| \leq r} \frac{|\r_n|}{\left|1 + \frac{\r_n}{k_n^o}\right|} \frac{1}{|k_n^o|^2}. 
	$$
	Since $\displaystyle \sup_{n \geq 1} \r_n = s < \iy$ and $|k_n^o| \to \iy$ as $n \to \iy$, we have 
	$\frac{|\r_n|}{|k_n^o|} \to 0$ as $n \to \iy$ and then there exists $C \in \R$ such that
	$$
	\frac{|\r_n|}{\left|1 + \frac{\r_n}{k_n^o}\right|} < C,\qq n \geq 1.
	$$
	By Corollary \ref{cor:im_part}, we have $\Im k_n^o \to \iy$ as $n \to \iy$. Hence, it follows from
	$\displaystyle \sum_{n \geq 1} \frac{|\Im k_n^o|}{|k_n^o|^2} < \iy$ that the series 
	$\displaystyle \sum_{n \geq 1} \frac{1}{|k_n^o|^2}$
	converges and then $\displaystyle \sum_{|k_n| \leq r} \frac{1}{|k_n^o|^2}$ is bounded as $r \to \iy$.
	Using these estimates, we see that 
	$\displaystyle \sum_{|k_n| \leq r} \left| \frac{1}{k_n} - \frac{1}{k_n^o} \right|$
	is bounded as $r \to \iy$, which yields that
	$\displaystyle \Bigl| \Re \sum_{|k_n| \leq r} \frac{1}{k_n} \Bigr|$ is bounded as $r \to \iy$.
	
	Now, it follows from Theorem \ref{thm:lind} that the function $f:\C \to \C$ given by
	$$
		f(k) = f(0) e^{ik\g} \lim_{r \to \iy} \prod_{|k_n| \leq r} \left( 1 - \frac{k}{k_n}\right),\qq k \in \C,
	$$
	is an entire function of exponential type, where 
	$\displaystyle f(0) = f_o(0) \lim_{r \to +\iy} \prod_{|k_n| \leq r} \frac{k_n}{k^o_n}$. 
	The last product converges, since
	$$
		\sum_{n \geq 1} \Bigl| 1 - \frac{k_n}{k^o_n} \Bigr| = \sum_{n \geq 1} \Bigl| \frac{\r_n}{k^o_n} \Bigr| \leq 
		\Bigl( \sum_{n \geq 1} |\r_n|^2 \Bigr)^{\frac{1}{2}} \Bigl(\sum_{n \geq 1} \frac{1}{|k_n^o|^2} \Bigr)^{\frac{1}{2}} \leq
		\sum_{n \geq 1} |\r_n| \Bigl(\sum_{n \geq 1} \frac{1}{|k_n^o|^2} \Bigr)^{\frac{1}{2}} < \iy.
	$$
	Here we used the H{\"o}lder inequality, $\r \in \ell^1 \ss \ell^2$ and
	$\displaystyle \sum_{n \geq 1} \frac{1}{|k_n^o|^2} < \iy$.

	Now, we show that $f \in \cW_+$. Since $f_o \in \cW_+$ and $f = \frac{f}{f_o} f_o$,
	it is sufficiently to show that $\frac{f}{f_o} \in \cW_+$. In order get this result, we introduce
	\[ \label{jost_pert:eq10}
		F(k) = \log \left( \frac{f(k)}{f_o(k)} \right),\qq k \in \R.
	\]
	Using the Hadamard factorization for $f$ and $f_o$, we have
	\[ \label{jost_pert:eq4}
		\begin{aligned}
			F(k) &= \log\left( \frac{f(0)}{f_o(0)} \lim_{r \to +\iy} \prod_{|k_n| \leq r} 
			\frac{1 - \frac{k}{k_n}}{1 - \frac{k}{k^o_n}} \right)
			= \log\left( \lim_{r \to +\iy} \prod_{|k_n| \leq r} \left(1 +  
			\frac{\r_n}{k^o_n - k} \right)\right)\\ 
			&= \lim_{r \to +\iy} \sum_{|k_n| \leq r} \log \left(1 +  
			\frac{\r_n}{k^o_n - k} \right) 
		\end{aligned}
	\]
	for any $k \in \R$. Now, we show that this series converges absolutely in $\hat \cL_+$. 
	By Lemma \ref{lm:wiener_estimate}, we have
	\[ \label{jost_pert:eq3}
		\left\| \frac{\r_n}{k^o_n - k} \right\|_{\hat \cL_+} = \frac{|\r_n|}{|\Im k^o_n|^{\frac{1}{2}}}\left(1 + \frac{1}{|\Im k^o_n|^{\frac{1}{2}}} \right).
	\]
	Recall that $|\Im k^o_n| \to \iy$ and $\r_n \to 0$ as $n \to \iy$. Thus, it follows from (\ref{jost_pert:eq3}) that
	$\left\| \frac{\r_n}{k^o_n - k} \right\|_{\hat \cL_+} \to 0$ as $n \to \iy$. Let $N \in \N$ be 
	such that $\left\| \frac{\r_n}{k^o_n - k} \right\|_{\hat \cL_+} < \frac{1}{4}$ for any $n > N$. 
	Thus, using Lemma \ref{lm:log_estimate} and estimate (\ref{jost_pert:eq3}), we obtain
	$$
		\left\| \log \left(1 + \frac{\r_n}{k^o_n - k} \right) \right\|_{\hat \cL_+} < \frac{C|\r_n|}{|\Im k^o_n|^{\frac{1}{2}}}\left(1 + \frac{1}{|\Im k^o_n|^{\frac{1}{2}}} \right)
	$$
	for any $n > N$, where the constant $C > 0$ does not depend on $n$.
	Using (\ref{jost_pert:eq4}), we get
	\begin{multline} \label{jost_pert:eq6}
		\sum_{n \geq 1} \Bigl\| \log \left(1 + \frac{\r_n}{k^o_n - k} \right) \Bigr\|_{\hat \cL_+} \leq\\
		\sum_{n = 1}^{N} \Bigl\| \log \left(1 + \frac{\r_n}{k^o_n - k} \right) \Bigr\|_{\hat \cL_+} + 
		\sum_{n > N} \frac{C|\r_n|}{|\Im k^o_n|^{\frac{1}{2}}}\left(1 + \frac{1}{|\Im k^o_n|^{\frac{1}{2}}} \right).
	\end{multline}
	Since $|\Im k^o_n| \to \iy$ as $n \to \iy$ and $|\Im k^o_n| \neq 0$ for any $n \geq 1$, 
	there exist a constant $C_1 > 0$ such that
	\[ \label{jost_pert:eq7}
		\frac{1}{|\Im k^o_n|^{\frac{1}{2}}}\left(1 + \frac{1}{|\Im k^o_n|^{\frac{1}{2}}} \right) < C_1
	\]
	for any $n \geq 1$. Due to Lemma \ref{lm:wiener_estimate} and Corollary \ref{p3c1}, 
	we have $\log \left(1 + \frac{\r_n}{k^o_n - k} \right) \in \hat \cL_+$ and then
	there exists a constant $C_2 > 0$ such that
	\[ \label{jost_pert:eq8}
		\Bigl\| \log \left(1 + \frac{\r_n}{k^o_n - k} \right) \Bigr\|_{\hat \cL_+} < C_2
	\]
	for any $1 \leq n \leq N$. Substituting (\ref{jost_pert:eq7}) and (\ref{jost_pert:eq8}) in (\ref{jost_pert:eq6}),
	we obtain
	\[ \label{jost_pert:eq9}
		\sum_{n \geq 1} \Bigl\| \log \left(1 + \frac{\r_n}{k^o_n - k} \right) \Bigr\|_{\hat \cL_+} \leq N C_2 + C C_1 \sum_{n > N} |\r_n|.
	\]
	Since $(\r_n)_{n \geq 1} \in \ell^1$, we have $\sum_{n > N} |\r_n| < \iy$ and then the series in (\ref{jost_pert:eq9})
	converges. Thus, we have $F \in \hat \cL_+$. Now, it follows from Corollary \ref{p3c1} that
	$\exp(F) \in \cW_{exp}$. Recall that $f_o \in \cW_{exp}$. 
	Thus, using (\ref{jost_pert:eq10}), we get $f = \exp(F) f_o \in \cW_{exp}$.
	
	Now, we show that $f \in \cJ$. Since $f \in \cW_{exp}$, it is bounded on $\R$ and then
	$$
		\int_{\R} \frac{\log(1 + |f(k)|) dk}{1 + k^2} < \iy.
	$$
	Recall that $f$ is an entire function of exponential type, which yields that $f \in \cE_{Cart}(\a,\b)$
	for some $\a,\b \in \R$.
	Moreover, we have $n(r,f) = \frac{2\g}{\pi} r + o(r)$ as $r \to \iy$. 
	Due to Theorem \ref{thm:lev}, we have $\t_+(f) + \t_-(f) = 2\g$. 
	Since $f \in \cW_+$, it follows from the Paley-Wiener theorem 
	that $\t_+(f) = 0$ and then we have $\t_-(f) = 2\g$. 
	Hence, we get $\t_+(f-1) = 0$ and $\t_-(f-1) = 2\g$. Moreover, we have showed above that
	$f - 1 \in \hat \cL_+$ and then, by the Plancherel theorem,
	we get $f - 1 \in L^2(\R)$. 
	Thus, using the Paley-Wiener theorem, we obtain
	$$
		f(k) - 1 = \int_0^{\g} g(s) e^{2iks} ds,\qq k \in \R,
	$$
	for some $g \in \cP$. Due to $f \in \cW_{exp}$, we have $f(k) \neq 0$ for any $k \in \ol \C_+$,
	 which yields that $f \in \cJ$.
	
	Finally, we show that $\r_{\cJ}(f,f_o) \to 0$ as $\| \r \|_{\ell^1} \to 0$. 
	If $\| \r_n \|_{\ell^1} < \ve$, then we have $|\r_n| < \ve$ for any $n \geq 1$.
	Using Lemma \ref{lm:wiener_estimate} and estimate (\ref{jost_pert:eq7}), we get
	$$
		\left\| \frac{\r_n}{k^o_n - k} \right\|_{\hat \cL_+} = \frac{|\r_n|}{|\Im k^o_n|^{\frac{1}{2}}}\left(1 + \frac{1}{|\Im k^o_n|^{\frac{1}{2}}} \right) < \frac{\ve}{|\Im k^o_n|^{\frac{1}{2}}}\left(1 + \frac{1}{|\Im k^o_n|^{\frac{1}{2}}} \right) < C_1 \ve
	$$
	for any $n \geq 1$ and then we get
	$$
		\left\| \frac{\r_n}{k^o_n - k} \right\|_{\hat \cL_+} < \frac{1}{4}
	$$
	for any $n \geq 1$ and $\ve < \frac{1}{4C_1}$. Thus, we can choose $N = 0$ in (\ref{jost_pert:eq9})
	for $\ve < \frac{1}{4C_1}$, which yields
	\[ \label{jost_pert:eq11}
		\left\| F \right\|_{\hat \cL_+} \leq 
		\sum_{n \geq 1} \| \log \left(1 + \frac{\r_n}{k^o_n - k} \right) \|_{\hat \cL_+} 
		\leq C C_1 \sum_{n > N} |\r_n| = C C_1 \| \r \|_{\ell^1}.
	\]
	It follows from Corollary \ref{p3c1}, that the mapping $F \mapsto \exp(F)$ from 
	$\hat \cL_+$ to $\cW_{exp}$ is continuous. Since $\exp(0) = 1$, we have $\|\exp(F) - 1\|_{\hat \cL_+} \to 0$ as
	$\left\| F \right\|_{\hat \cL_+} \to 0$ and then it follows from inequality 
	(\ref{jost_pert:eq11}) that
	$\|\exp(F) - 1\|_{\hat \cL_+} \to 0$ as $\| \r \|_{\ell^1}$. 
	Now, we consider $\r_{\cJ}(f,f_o)$. By the definition of
	the metric $\r_{\cJ}$, we have
	$$
		\r_{\cJ}(f,f_o) = \|f-f_o\|_{\hat \cL_+} = \|\exp(F) f_o - f_o\|_{\hat \cL_+} \leq \|f_o\|_{\cW_+} \|\exp(F) - 1\|_{\hat \cL_+}. 
	$$
	Thus, it follows from this inequality that $\r_{\cJ}(f,f_o) \to 0$ as 
	$\| \r \|_{\ell^1} \to 0$.
\end{proof}
\begin{proof}[\textbf{Proof of Theorem \ref{thm:main_thm}}]
	Let $\k^o = (k_n^o)_{n \geq 1}$ be zeros of $\psi_o = \psi(\cdot,q_o)$ for some $q_o \in \cP$ 
	arranged that $0 < |k^o_1| \leq |k^0_2| \leq \ldots$ and 
	let $\r = (\r_n)_{n \geq 1} \in \ell^1$ be such that $k_n = k_n^o + \r_n \in \C_-$ 
	for each $n \geq 1$. Due to Theorem \ref{thm:inv_jost}, we have $\psi_o \in \cJ$ and then, 
	by Theorem \ref{thm:jost_pert}, there exists a unique $\psi_1 \in \cJ$ such that 
	$(k_n)_{n \geq 1}$ are zeros of $\psi_1$. Now, it follows from Theorem \ref{thm:inv_jost} that
	there exists a unique $q \in \cP$ such that $\psi_1 = \psi(\cdot,q)$. Moreover, if $\r_{\cJ}(\psi_1,\psi_o) \to 0$
	then we have $\r_{\cP}(q,q_o) \to 0$. Due to Theorem \ref{thm:jost_pert}, we have 
	$\r_{\cJ}(\psi_1,\psi_o) \to 0$ as $\|\r\|_{\ell^1} \to 0$, which yields $\r_{\cP}(q,q_o) \to 0$.
\end{proof}

\begin{proof}[\textbf{Proof of Theorem \ref{thm:herm_pert}}]
	Let $E_o$ be Dirac-type Hermite-Biehler function, let 
	$\k^o = (k_n^o)_{n \geq 1}$ be its zeros arranged that $0 < |k^o_1| \leq |k^0_2| \leq \ldots$ and 
	let $\r = (\r_n)_{n \geq 1} \in \ell^1$ be such that $k_n = k_n^o + \r_n \in \C_-$ 
	for each $n \geq 1$. We introduce $\psi_o(k) = i e^{i\g k} E_o(k)$, $k \in \C$. 
	Due to Theorem \ref{thm4}, the function $\psi \in \cJ$ and then, by Theorem \ref{thm:jost_pert},
	there exists a unique $\psi \in \cJ$ such that $(k_n)_{n \geq 1}$ are zeros of $\psi$ and
	$\r_{\cJ}(\psi,\psi_o) \to 0$ as $\|\r\|_{\ell^1} \to 0$. Thus, it follows from Theorem \ref{thm4}
	that there exists a unique Dirac-type Hermite-Biehler function $E$ such that 
	$\k = (k_n)_{n \geq 1}$ are zeros of $E$ and given by $E(k) = -ie^{-i\g k}\psi(k)$, $k \in \C$.

	Now, we show the continuity. Using well-known properties of the Fourier transform, we get
	$$
		\|\cF^{-1}(E - E_o)\|_{L^2(-\frac{\g}{2},\frac{\g}{2})} = 
		\|\cF^{-1}(e^{-i g k} (\psi - \psi_o))\|_{L^2(-\frac{\g}{2},\frac{\g}{2})} =
		\|\cF^{-1}(\psi - \psi_o)\|_{L^2(0,\g)} = \r_{\cJ}(\psi,\psi_o), 
	$$
	which yields that $\|\cF^{-1}(E - E_o)\|_{L^2(-\frac{\g}{2},\frac{\g}{2})} \to 0$ as $\|\r\|_{\ell^1} \to 0$.
\end{proof}

\footnotesize

\no {\bf Acknowledgments.} D. M. is supported by the RFBR grant No. 19-01-00094.

\end{document}